\documentclass{article}

\usepackage[utf8]{inputenc}					

\usepackage[margin=1in]{geometry}		

\usepackage{amsmath,amsthm,amsfonts}		
\usepackage{dsfont}
\usepackage{mathtools}
\usepackage[usenames,dvipsnames]{xcolor}

\usepackage[inline, shortlabels]{enumitem}  

\usepackage[linesnumbered,ruled,vlined]{algorithm2e}
\RestyleAlgo{boxruled}

\usepackage{graphicx}						

\usepackage{hyperref}						
\hypersetup{
	hidelinks,
    colorlinks = true,
    linkcolor = MidnightBlue,
    citecolor = MidnightBlue
}

\usepackage[noabbrev, capitalize,nameinlink]{cleveref}			

\makeatletter
\newtheorem*{rep@theorem}{\rep@title}
\newcommand{\newreptheorem}[2]{%
\newenvironment{rep#1}[1]{%
 \def\rep@title{#2 \ref{##1} (restated)}%
 \begin{rep@theorem}}%
 {\end{rep@theorem}}}
\makeatother

\newtheorem{theorem}{Theorem}
\newreptheorem{theorem}{Theorem}

\newtheorem{lemma}[theorem]{Lemma}

\newtheorem{claim}[theorem]{Claim}

\newtheorem{corollary}[theorem]{Corollary}

\theoremstyle{definition}
\newtheorem{definition}[theorem]{Definition}

\crefname{claim}{Claim}{Claims}
\crefname{observation}{Observation}{Observations}
\crefname{equation}{Eq.}{Eqs.}

\DeclareMathOperator{\bbE}{\mathbb{E}}
\DeclareMathOperator{\bbP}{\mathbb{P}}
\DeclareMathOperator{\bbN}{\mathbb{N}}
\DeclareMathOperator{\bbR}{\mathbb{R}}

\newcommand{\bbOne}{\mathds{1}}
\newcommand{\pa}[1]{\left( #1 \right)}

\newcommand*{\medcap}{\mathbin{\scalebox{1.5}{\ensuremath{\cap}}}}

\newcommand{\argmax}{\textrm{argmax}}

\newcommand{\E}{{\rm I\kern-.3em E}}
\newcommand{\Var}{\mathrm{Var}}

\newcommand{\dis}{\mathrm{dis}}

\title{Minimalist Leader Election Under Weak Communication}

\author{Robin Vacus\thanks{Humboldt University of Berlin, Germany. E-mail: \texttt{robin.vacus@hu-berlin.de}. Research supported in part by the Deutsche Forschungsgemeinschaft (DFG, German Research Foundation) – project number 539576251. } \and Isabella Ziccardi\thanks{Université Paris Cité, CNRS, IRIF, Paris, France. E-mail: \texttt{isabella.ziccardi@irif.fr}.  Research supported in part by the European QuantERA project QOPT (ERA-NET Cofund 2022-25) and the French PEPR integrated project EPiQ (ANR-22-PETQ-0007).}}

\date{}

\usepackage{tikz}
\usepackage{tablefootnote}

\definecolor{myred}{RGB}{200, 0, 0}

\newcommand{\lead}{\bullet}
\newcommand{\nlead}{\circ}

\newcommand{\Prob}[1]{\bbP\left(#1\right)}
\newcommand{\Expc}[1]{\mathbb{E}\left(#1\right)}

\usepackage[skip = 5pt, indent=10pt]{parskip}

\newcommand{\impliesplus}[1]{\underset{\text{\Cref{#1}}}{\implies}}
\newcommand{\beepcount}{\mathcal{N}^{\mathrm{beep}}}
\newcommand{\master}{\mathcal{M}^\star}

\usepackage{bbding}
\usepackage{pifont}
\usepackage{subcaption}

\newcommand{\protocolname}{BFW\xspace}

\begin{document}

\maketitle 


\begin{abstract}
    We propose a protocol to solve Leader Election within weak communication models such as the beeping model or the stone-age model.
    Unlike most previous work, our algorithm operates on only six states, does not require unique identifiers, and assumes no prior knowledge of the network's size or topology, i.e., it is uniform. We show that under our protocol, the system almost surely converges to a configuration in which a single node is in a leader state. With high probability, this occurs in fewer than $O(D^2 \log n)$ rounds, where $D$ is the network diameter. We also show that this can be decreased to $O(D \log n)$ when a constant factor approximation of~$D$ is known.
    The main drawbacks of our approach are a $\Tilde{\Omega}(D)$
    overhead in the running time compared to algorithms with stronger requirements, and the fact that nodes are unaware of when a single-leader configuration is reached. Nevertheless, the minimal assumptions and natural appeal of our solution make it particularly well-suited for implementation in the simplest distributed systems, especially biological ones.
\end{abstract}




\section{Introduction}

Real distributed systems exhibit widely varying levels of sophistication, and a key challenge in distributed computing is to develop theoretical models suited to each.
For instance, highly complex artificial systems benefit from large bandwidth and substantial computing power, whereas primitive biological entities such as bacteria and insects, or the tiniest robots, must function with limited cognition and communication. Over the past few decades, several models have been designed specifically to study the latter, including population protocols~\cite{AspnesR09}, radio networks~\cite{chlamtac_broadcasting_1985}, the stone-age model~\cite{emek_stone_2013}, and the beeping model~\cite{CornejoK10}. These models share the goal of identifying protocols that efficiently perform core distributed tasks (e.g., broadcast, leader election, MIS) while making minimal assumptions about the agents' capabilities.
However, despite the strict constraints of these models, many algorithms that have been proposed remain relatively complex with respect to biological motivations,
as the focus is generally on performance rather than simplicity.

Consider, for example, the leader election problem in the beeping model of communication, which will be described in \Cref{sec:problem_def}.
As we will see in the dedicated section (\Cref{sec:related_works}), state-of-the-art algorithms in this setting often rely on unique agent identifiers (or methods to generate them)~\cite{DufoulonBB18,ForsterSW14,GhaffariH13,CzumajD19}, and manage to communicate them across the group; something simple organisms like bacteria are unlikely to achieve. Beyond the biological plausibility issue, this also introduces non-trivial technical assumptions: (1) agents must have at least an approximate knowledge of the group size, $n$, and (2) the amount of memory used by the protocol is unbounded as~$n$ increases.
These assumptions are not made lightly; they are necessary to enable the design of efficient, non-trivial algorithms with the best possible properties.
For example, an approximate knowledge of~$n$ is known to be necessary in order to ensure that agents can identify when a leader has been elected~\cite{itai_symmetry_1990}.

However, in a quest for simplicity, some articles restrict the agents' capabilities even further. This approach may involve, for example, limiting the number of memory states accessible to the agents by a constant quantity, which, although arbitrarily large, remains independent of $n$. It can also require the algorithms to be uniform, meaning they do not depend on the communication graph $G$, including its size and diameter.
At the extreme, agents are modeled as probabilistic finite-state machines that are anonymous, identical, and independent of~$G$. Among other things, this prohibits the use of unique identifiers (even when generated by the protocol itself), as agents lack sufficient memory to store them. While such rules are relatively standard in population protocols, they are more rarely adopted in the beeping model, especially the restriction on the number of states. Notable exceptions include~\cite{giakkoupis_distributed_2023} for MIS and~\cite{gilbert_computational_2015} for leader election in the clique.
In this paper, we build upon these works by presenting a simple and natural protocol, described in \Cref{sec:algorithm}, that solves leader election in any graph under the beeping model, while satisfying all the aforementioned restrictions.
Our algorithm can also be implemented in a synchronous version of the stone-age model.

\subsection{Problem Definition} \label{sec:problem_def}

We consider an undirected connected graph $G = (V , E)$ of~$n$ nodes, where there is an edge between two nodes if they can communicate with each-other.
Execution proceeds in discrete rounds. In each round, each node must decide whether to {\em beep} or remain silent ({\em listen}). A node that chooses to listen in a given round will {\em hear} a beep if and only if at least one of its neighbors beeps in the same round. As a consequence, a node cannot differentiate between a single neighbor emitting a beep or multiple neighbors doing so simultaneously.

In this paper, using a similar terminology as in~\cite{gilbert_computational_2015}, a protocol is defined as a probabilistic state machine~$M = (Q_\ell, Q_b, q_s, \delta_\bot, \delta_\top)$ where: $Q_\ell$ and~$Q_b$ are two disjoint sets of states corresponding to listening and beeping, respectively; $q_s$ is the initial state; and $\delta_\bot, \delta_\top$ are probabilistic transition functions, mapping the current state of a node to a distribution over states to enter in the next round.
Specifically, if a node~$u$ is in a beeping state in round~$t$, or if any of its neighbor is, then the next state of~$u$ is sampled according to~$\delta_\top$; otherwise it is given by~$\delta_\bot$.

We focus on {\em uniform} protocols which are independent of~$G$. A direct consequence is that node cannot maintain unique identifiers. Moreover, positive transition probabilities in~$\delta_\bot, \delta_\top$ being fixed with respect to~$n$, they cannot tend to~$0$ as~$n$ increases.

\begin{definition}[Eventual Leader Election] \label{def:eventual_LE}
A protocol in the Beeping model is said to solve \textit{Eventual Leader Election} in time~$T$ if there exists a subset of states $L \subset Q_\ell \cup Q_b$ and a node $u^\star$ s.t. from every round~$t \geq T$, $u^\star$ is the only node whose state belongs to~$L$ in round~$t$.
\end{definition}

Note that this definition of Leader election does not require nodes to commit to a final state; it only requires that a configuration with a single leader is reached and maintained indefinitely (hence the term {\em eventual}). In particular, it means that nodes may never be aware of whether the protocol has converged. This is in contrast with many related works such as~\cite{DufoulonBB18}, which often impose stronger termination conditions (see \Cref{sec:related_works}).
However, our assumption that agents behave independently of~$n$ prevents them from identifying the moment a leader is elected, as noted in~\cite{itai_symmetry_1990}. This justifies our choice of \Cref{def:eventual_LE}.





\subsection{Our Algorithm: \protocolname} \label{sec:algorithm}

\paragraph{Informal description.} At the start of execution, each node is initialized as a leader. 
A leader that does not hear a beep, beeps in the next round with probability~$p > 0$, where~$p \in (0,1)$ may be any constant independent of~$n$ (say~$1/2$).
Non-leader nodes remain silent unless they hear a beep, in which case they beep in the next round.
After each beep, a node (leader or non-leader) becomes ``frozen'' for exactly one round, during which it does not beep and does not react to its environment.
A leader that hears a beep while not being frozen is immediately ``eliminated'', i.e., it becomes a non-leader -- in which case it beeps in the next round.

More precisely, our protocol operates on six states: $\{B^\nlead, F^\nlead, W^\nlead\}$ are used by non-leader nodes, and $\{B^\lead, F^\lead, W^\lead\}$ are reserved for leader nodes, where $B$ stands for {\em Beeping}, $F$ for {\em Frozen}, and~$W$ for {\em Waiting}.
The exact transitions between states are depicted on \Cref{fig:WBF}.

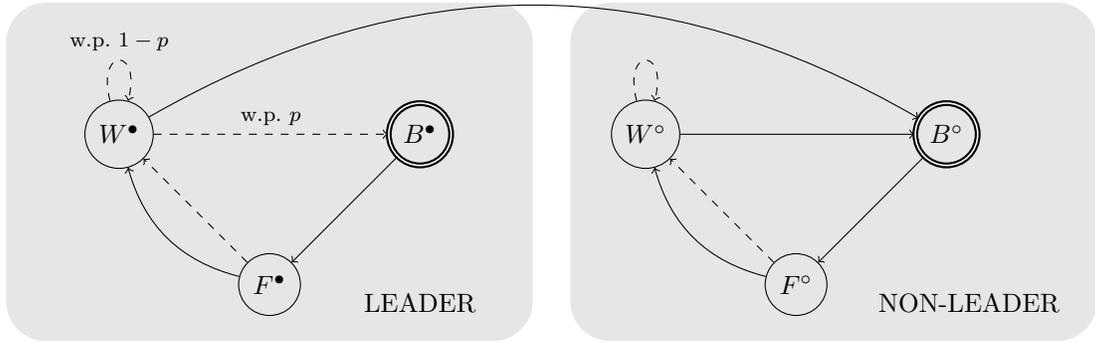
\begin{figure} [htbp]
    \centering
    \begin{tikzpicture}[
    beeping/.style = {draw, double, circle, thick},
    listening/.style = {draw, circle},
    edge_false/.style = {dashed,->},
    edge_true/.style = {->},
    invisible/.style={minimum width=0mm,inner sep=0mm,outer sep=0mm}
    ]

    \tikzset{invisible/.style={minimum width=0mm,inner sep=0mm,outer sep=0mm}}
    
    \fill[rounded corners=15pt, fill=gray!20] (-1.5,8) rectangle ++(7,-4.5);
    \node at (4,4) {LEADER};
    \fill[rounded corners=15pt, fill=gray!20] (6,8) rectangle ++(7,-4.5);
    \node at (11.3,4) {NON-LEADER};
        
	\node[listening] (1) at (0,6.25) {$W^\lead$};
	\node[beeping] (2) at (4,6.25) {$B^\lead$};
	\node[listening] (3) at (2,4.25) {$F^\lead$};
	\node[listening] (4) at (7,6.25) {$W^\nlead$};
	\node[beeping] (5) at (11,6.25) {$B^\nlead$};
	\node[listening] (6) at (9,4.25) {$F^\nlead$};

	\draw[edge_false,loop above] (1) to node[midway,above,align=center] {\footnotesize w.p.~$1-p$} (1) ;
	\draw[edge_false] (1) to node[midway,above,align=center] {\footnotesize w.p.~$p$} (2) ;
	\draw[edge_true] (2) to (3) ;
	\draw[edge_false] (3) to (1) ;
    \draw[edge_true,bend left] (3) to (1) ;
    \draw[edge_true,bend left] (1) to (5) ;
	\draw[edge_false,loop above] (4) to (4) ;
	\draw[edge_true] (4) to (5) ;
	\draw[edge_true] (5) to (6) ;
	\draw[edge_false] (6) to (4) ;
    \draw[edge_true,bend left] (6) to (4) ;
\end{tikzpicture}
    \caption[Definition of our algorithm as a probabilistic finite-state machine]{{\bf Definition of Algorithm} \protocolname { \bf as a probabilistic finite-state machine.}
    The starting state is~$W^\lead$.
    Solid lines indicate transitions corresponding to~$\delta_\top$, which occur when a beep is heard, while dashed lines indicate transitions corresponding to~$\delta_\bot$, which occur when both the node and its neighborhood are listening.
    A node beeps if its state belongs to~$Q_b = \{B^\lead, B^\nlead\}$, which are circled twice in the figure\footnotemark.
    A node is considered as a leader if its state belongs to~$\{B^\lead, F^\lead, W^\lead\}$, which make up the first half of the figure.
    All transitions are deterministic except for~$\delta_\bot(W^\lead)$.
    }
    \label{fig:WBF}
\end{figure}

\footnotetext{Transitions corresponding to~$\delta_\bot$ are not specified for~$\{B^\lead, B^\nlead\}$, since by definition, $\delta_\top$ is systematically used from these states.}

\subsection{Our Results}

We write~$n$ to denote the size of the graph, and~$D$ its diameter.
Our main result is that when~$p$ is fixed independently of~$n$, Algorithm \protocolname always elects a single leader, and we give an upper bound on its convergence time that holds with high probability\footnote{We say that an event holds \emph{with high probability} (w.h.p.) if it holds with probability $1-n^{-\Omega(1)}$, and \emph{almost surely} if it holds with probability~$1$.}.

\begin{theorem} \label{thm:main}
    Fix~$p \in (0,1)$.
    Algorithm \protocolname with parameter~$p$, defined in \Cref{fig:WBF}, solves Eventual Leader Election on any graph $G=(V,E)$ almost surely\footnotemark[2], and elects a single leader in less than~$O(D^2 \log n)$ rounds with high probability. 
    \label{thm:convergence}
\end{theorem}

\paragraph{Beep waves.}

By following our protocol, agents collectively create {\em beep waves} that expand away from leaders. Leaders crossed by a wave are eliminated, and two waves emitted by different leaders ``crash'' against each-others and disappear. As long as leaders beep with different frequencies (which is achieved via randomization), they are gradually eliminated, until only one remains.
Beep waves are a common method for transmitting information over long distances in the beeping model~\cite{GhaffariH13,CzumajD19},
where they are typically used to convey messages bit by bit. In contrast, our approach is minimalist, as nodes do not need to (and cannot) monitor the specific wave pattern received over an extended period of time.

\paragraph{Ideas behind the analysis.}

The main difficulty is to show that not all leaders are eliminated. This requires proving that waves never return to their origin, so that a leader~$u$ may not eliminate itself. More generally, this requires ruling out any scenario involving a set of leaders~$v_1,\ldots,v_k$ where a wave from~$v_i$ eliminates~$v_{i+1}$, and a wave from~$v_k$ eliminates~$v_1$.
To achieve this, we introduce a quantity that we call {\em flow} (\Cref{def:paths}), because it shares some of the properties commonly associated with this term.
Informally, the flow over a given oriented path in round~$t$ measures the difference between the number of waves traveling from its start to its end, and those traveling in the opposite direction. By construction of our algorithm, the flow between two leaders can be directly determined from the number of beeps each leader has emitted so far (see \Cref{cor:flow_path}). This property enables us to navigate the apparent complexity of the process, and the result then follows from probabilistic arguments.
Considerations about the flow are fully deterministic, and are discussed in \Cref{sec:flow}, while probabilistic proofs are presented in \Cref{sec:probabilistic_analysis}. The proof of \Cref{thm:main} itself can be found in \Cref{sec:main_proof}.

\paragraph{On the use of randomness.} 

While some effort have been devoted to developing deterministic algorithms~\cite{ForsterSW14,DufoulonBB18}, our protocol requires a certain degree of randomness to function. This randomness is, in fact, strictly necessary in our setting, as no deterministic algorithm can break symmetry on certain graphs when all agents are identical.
Moreover, we make a parsimonious use of random bits, since agents only need one per round when~$p=\frac{1}{2}$.



\paragraph{Faster convergence in the non-uniform case.} For the sake of completeness, we identify a value of~$p$ that yields a better convergence time, closer to what can be achieved by more sophisticated algorithms (see \Cref{sec:related_works}). However, its dependence on~$D$ makes the protocol non-uniform. In addition, this value might become arbitrarily small as~$D$ increases, necessitating a large number of random bits in every round.


\begin{theorem} \label{thm:side}
    Algorithm \protocolname with parameter~$p = \frac{1}{D+1}$ solves Eventual Leader Election in less than~$O(D \log n)$ rounds with high probability.
\end{theorem}

Although the theorem features an exact value for the parameter~$p$, its proof can be easily generalized to the case that only a constant factor approximation of~$D$ is known. The proof of \cref{thm:side} can be found in \cref{sec:side_result}.


\subsection{Related Works} \label{sec:related_works}

\paragraph{Other algorithms for Leader Election in the Beeping model.}

Other works typically consider
a stronger definition of Leader Election. They often enforce a {\em safety} condition, ensuring that the population never contains more than one node in a leader state. Moreover, they require nodes to {\em terminate}, i.e., to commit on their final state. Some works also require nodes to detect when a final configuration has been reached, and even the identity of the elected leader.

Importantly however, these works do not attempt to minimize the number of states available to the agents, with the notable exception of~\cite{gilbert_computational_2015}, which is discussed separately below. To the best of our knowledge, all others rely on storing integers as large as~$\Theta(\log n)$, resulting in~$\Theta(n)$ different states. This relatively large amount of memory allows for more sophisticated algorithms, and is crucial to satisfying the stronger requirements listed above.
In terms of running time, the algorithm in~\cite{GhaffariH13} converges in~$O(D + \log n \log \log n) \cdot \min \{ \log \log n, \log \frac{n}{D} \}$ rounds with high probability.
The same paper states a~$\Omega(D + \log n)$ lower bound, applicable to both deterministic and probabilistic protocols.
A subsequent work~\cite{ForsterSW14} presents a deterministic algorithm converging in~$O(D \log n)$ rounds.
Since then, more efficient algorithms (both random and deterministic) have been identified that reach the lower bound: for example, the protocol in ~\cite{DufoulonBB18} converges in~$O(D + \log n)$ rounds, and the one in~\cite{CzumajD19} runs in time proportional to the broadcasting time. A comparison of the performance of these algorithms and the protocol presented in this paper is provided in \Cref{tab:result}.

In contrast, the protocols in~\cite{gilbert_computational_2015} follow the same simplicity criteria considered here, while addressing a stronger variant of the problem than ours. Specifically, they enforce the aforementioned safety condition, and also prohibit nodes from leaving the leader state. However, unlike our setting, they allow for an error probability~$\varepsilon$, which affects both their convergence time and the number of states used by a factor of approximately~$\log \frac{1}{\varepsilon}$. Moreover, these protocols are limited to single-hop (i.e., fully connected) networks.
Our work can be seen as an extension to arbitrary graphs, using a somewhat simpler, albeit less efficient, algorithm.

\paragraph{Radio Networks and the Stone Age model.} Leader Election has also received some attention within Radio Networks~\cite{chlamtac_broadcasting_1985} and the Stone Age model~\cite{emek_stone_2013}, which are closely related to the Beeping model.
In Radio Networks, messages are usually made up of~$\Theta(\log n)$ bits. However, they are received by a node only when exactly one neighbor is emitting in a given round; when several neighbors are emitting simultaneously, the content of the message is lost, and the collision might or might not be detected depending on the setting.
The leader election problem itself is discussed in, e.g.,~\cite{chlebus_electing_2012,kowalski_leader_2009}.
In the Stone Age model, nodes are activated asynchronously, and communicate by displaying messages within an alphabet of finite size. When activated, a node~$u$ can count the number of neighbors holding the same message~$\sigma$, but only up to some threshold~$b \geq 1$.
Importantly, both models allow nodes to accurately detect the situation where a single neighbor emits a signal (unless~$b=1$), which significantly impacts algorithm design.
Leader election in the stone-age model has been addressed in \cite{EmekK21}, where a randomized self-stabilizing solution is proposed, with state space $O(D)$ and stabilizing in $O(D \log n)$ interactions w.h.p. The algorithm does not assume unique IDs or any knowledge of $n$, but assumes the knowledge of $D$, and has no termination detection.

\paragraph{Population protocols.} 

Leader election has been studied in the Population Protocols model~\cite{AspnesR09}. In the classical definition of this model, an edge is chosen uniformly at random at each time step, after which the corresponding nodes interact and may update their states.
As in this paper, leader election in Population Protocols is typically understood as eventual, meaning that detecting termination is not required. Most works assume that the communication graph is fully connected.
On the clique, leader election requires at least $\Omega(n^2)$ expected interactions with constantly many states~\cite{DotyS18},
$\Omega(n^2/\mathrm{polylog}(n))$ expected interactions with~$O(\log \log n)$ states~\cite{AlistarhAEGR17}, and~$\Omega(n \log n)$ when the number of states is unbounded~\cite{SudoM20}. The protocol in~\cite{BerenbrinkGK20} achieves optimal space and time complexity, using $O(\log \log n)$ states and $O(n \log n)$ expected interactions.
For arbitrary graphs, a $O(\log^2 n)$ state algorithm has been recently identified~\cite{alistarh_nearoptimal_2022}, whose convergence time is~$O(\text{Broadcast time} \cdot \log n)$.
Overall, although the Beeping Model and Population Protocols share similar motivations, they exhibit significant differences that make it difficult to compare convergence times across the two settings.

\begin{table}[t]
\centering
\caption{Overview of existing results regarding Leader Election in the Beeping model.}
\label{tab:result}
\begin{tabular}{c|c|c|c|c|c|c}
Ref.                                  & Round Complexity                                                                                                   & Unique IDs & Knowledge & Safety & States &\begin{tabular}[c]{@{}c@{}}Termination \\ Detection\end{tabular} \\ \hline
\cite{GhaffariH13}   & \begin{tabular}[c]{@{}c@{}}$O(D+ \log n \log \log n)$\\ $\cdot \min\{\log \log n, \log \frac{n}{D}\}$\end{tabular} & yes        & $n,D$     & w.h.p. & $\Omega(n)$ & yes 
\\ \cite{ForsterSW14} & $O(D \log n)$ & yes & none & det. & $\Omega(n)$ & yes

\\
\cite{DufoulonBB18}  & $O(D + \log n)$                                                                                                    & yes        & none      & det.   & $\Omega(n)$ & yes                                                              \\
\cite{DufoulonBB18} & $O(D+\log n)$                                                                                                      & no         & $n$       & w.h.p. & $\Omega(n)$ & yes                                                              \\
\cite{EmekK21} & $O(D \log n)$  & no & $D$ & w.h.p.\tablefootnote{The algorithm is self-stabilizing and uses a constant number of channels.} & $\Omega(D)$ & no
\\
\textbf{This Paper}                            & $O(D^2 \log n)$                                                                                                    & no         & none      & w.h.p. & $O(1)$ & no                                                               \\
\textbf{This Paper}                            & $O(D \log n)$                                                                                                      & no         & $D$       & w.h.p. & $O(1)$ & no\tablefootnote{In this case, assuming the additional knowledge of~$n$, the algorithm could stop after~$\Omega(D \log n)$ rounds to achieve termination detection w.h.p.}                                                             
\end{tabular}
\end{table}


\section{Notations and Definitions}

For~$d \in \{0,\ldots,D\}$, we write~$N_d(u)$ to denote the~$d$-neighborhood of~$u$, that is, the set of nodes which are at distance exactly~$d$ from~$u$.

For any state in~$\{W^\nlead,B^\nlead,F^\nlead,W^\lead,B^\lead,F^\lead\}$, we add a subscript~$t$ to denote the set of nodes that are in this state in round~$t$. Moreover, we write~$B_t = B^\nlead_t \cup B^\lead_t$ the set of all beeping nodes (leaders and non-leaders), and similarly, we write~$W_t = W^\nlead_t \cup W^\lead_t$ and~$F_t = F^\nlead_t \cup F^\lead_t$.
For example, with this notation, we write ``$u \in B_t$'' to express the fact that~$u$ is beeping in round~$t$.

We write~$\beepcount_t(u)$ the number of times that~$u$ beeped up to round~$t$ (included):
\begin{equation*}
    \beepcount_t(u) := \mathrm{card} \{ s \leq t, u \in B_s \}.
\end{equation*}


\begin{definition} [Paths] \label{def:paths}
    A {\em path} is a sequence of oriented\footnote{This definition of a path using oriented edges does not contradict our assumption that~$G$ is an undirected graph; rather, we simply require that each oriented edge in the path corresponds to an undirected edge in~$G$ with the same endpoints.} edges~$(e_1,\ldots,e_k)$ for which there is a sequence of vertices~$(v_1,\ldots,v_{k+1})$ such that for every~$i\in\{1,\ldots,k\}$, $e_i = (v_i,v_{i+1})$.
    We refer to~$(v_1,\ldots,v_{k+1})$ as the {\em vertex sequence} of~$\omega$.
    Importantly, we do not require edges and vertices in a path to be distinct.
\end{definition}

\begin{definition} [Flow] \label{def:flow}
    The {\em flow} along an oriented edge~$e = (u,v)$ in round~$t$ is defined as
    \begin{equation*}
        \nu_t(e) = \begin{cases}
            +1 & \text{if } u \in B_t \text{ and } v \in W_t, \\
            -1 & \text{if } u \in W_t \text{ and } v \in B_t, \\
            0 & \text{otherwise.}
        \end{cases}
    \end{equation*}
    Similarly, the flow can be defined along a path~$\omega = (e_1,\ldots,e_k)$ as
    \begin{equation*}
        \nu_t(\omega) = \sum_{i=1}^{k} \nu_t(e_i).
    \end{equation*}
    Note that as a consequence of the definition, $|\nu_t(e)| \leq 1$ and
    \begin{equation} \label{eq:max_flow}
        |\nu_t(\omega)| \leq k.
    \end{equation}
\end{definition}

\section{Deterministic Properties of the Flow} \label{sec:flow}

In this section, we establish a few properties of our protocol (\Cref{fig:WBF}) based on our notion of flow (\Cref{def:flow}).
We assume that all nodes start in a Waiting state, and that there is at least one leader in round~$0$, or in other words:
\begin{equation} \label{eq:initialization_condition}
    V \subseteq W_0 \quad \text{ and } \quad W_0^\lead \neq \emptyset.
\end{equation}
In the following claim, we list all the conditions on which the other results of the section rely (in addition to \Cref{eq:initialization_condition}).
Interestingly, none of these conditions depend on the rule governing when to go from~$W^\lead$ to~$B^\lead$ (which is a simple coin toss in our case). In fact, all the results in this section hold deterministically and independently of this rule.
\begin{claim} [Basic Observations] \label{claim:transitions}
    Let~$\{u,v\} \in E$.
    For every~$t\geq 0$,
    \begin{align}
        u \in W_t &\implies u \notin F_{t+1}, \label{eq:W_next} \\
        u \in B_t &\implies u \in F_{t+1}, \label{eq:B_next} \\
        u \in F_t &\implies u \in W_{t+1}, \label{eq:F_next} \\
        u \in B_t \text{ and } v \in W_t &\implies v \in B^\nlead_{t+1}, \label{eq:WB_next}
    \end{align}
    and for every~$t>0$,
    \begin{align}
        u \in W_t &\implies u \notin B_{t-1}, \label{eq:W_previous} \\
        u \in B_t &\implies u \in W_{t-1}, \label{eq:B_previous} \\
        u \in F_t &\implies u \in B_{t-1}, \label{eq:F_previous} \\
        u \in F_t \text{ and } v \in W_t &\implies v \in F_{t-1}, \label{eq:WF_previous} \\
        u \in B_t^\nlead &\implies N_1(u) \cap B_{t-1} \neq \emptyset. \label{eq:B_nlead_previous}
    \end{align}
\end{claim}
\begin{proof}
    \Cref{eq:W_next,eq:B_next,eq:F_next,eq:W_previous,eq:B_previous,eq:F_previous,eq:B_nlead_previous} are direct consequences of the transitions of the protocol depicted in \Cref{fig:WBF}.
    
    If~$u \in B_t$ and~$v \in W_t$, then~$v$ hears a beep in round~$t$, and so~$v \in B^\nlead_{t+1}$ by the transitions of the protocol, which establishes \Cref{eq:WB_next}.
    
    If~$u \in F_t$ and~$v \in W_t$, then~$u \in B_{t-1}$ by \Cref{eq:F_previous} and~$v \in W_{t-1} \cup F_{t-1}$ by \Cref{eq:W_previous}.
    However, we cannot have~$v \in W_{t-1}$ (otherwise we would have~$v \in B_t$ by \Cref{eq:WB_next}), therefore~$v \in F_{t-1}$, which establishes \Cref{eq:WF_previous} and concludes the proof of \Cref{claim:transitions}.
\end{proof}

Next, we show that the evolution of the flow along a path~$\omega$ from one round to the next, only depends on which endpoints of~$\omega$ are beeping.

\begin{lemma} [Conservation of flow] \label{lemma:flow_conservation}
    Let~$\omega$ be a path with vertex sequence~$(v_1,\ldots,v_k)$, and let~$t > 0$. We have
    \begin{equation*}
        \nu_{t}(\omega) = \nu_{t-1}(\omega) + \bbOne\{v_1 \in B_t\} - \bbOne\{v_k \in B_t\}.
    \end{equation*}
\end{lemma}

\begin{proof}
    Let~$\omega$ be a path with vertex sequence~$(v_1,\ldots,v_k)$ and a round~$t \geq 0$.
    We consider the word~$w$ of length~$k$, written on the alphabet~$\{W,B,F\}$, such that the~$i^{\text{th}}$ letter of~$w$ corresponds to the state of~$v_i$ in round~$t$. Given a regular expression~$e$, we say that~$\omega$ {\em matches}~$e$ in round~$t$, and write~$\omega \sim_t e$, if~$w \in \mathcal{L}(e)$.
    Throughout the proof, we will use the same notation for sub-paths of~$\omega$.

    If~$\omega \sim_t B^k$, then~$\omega \sim_{t-1} W^{k}$ by \Cref{eq:B_previous}. Therefore, $\nu_{t}(\omega) = \nu_{t-1}(\omega) = 0$ and the statement holds.
    Otherwise, let~$m = \min \{i, v_i \notin B_t \}$, and~$M = \max \{i, v_i \notin B_t \}$.
    We consider the following three sub-paths of~$\omega$: $\omega_1 = (v_1,\ldots,v_m)$, $\omega_2 = (v_m,\ldots,v_M)$ and $\omega_3 = (v_M,\ldots,v_k)$.
    \begin{itemize}
        \item If $\omega_1 \sim_t B^{m-1}W$, then $\omega_1 \sim_{t-1} W^{m-1} \cdot (W+F)$ by \Cref{eq:B_previous,eq:W_previous}. In that case, $\nu_{t-1}(\omega_1) = 0$, and
        \begin{equation*}
            \nu_t(\omega_1) = \begin{cases}
                1 & \text{if } m \geq 2, \\
                0 & \text{otherwise}
            \end{cases} \quad = \bbOne\{v_1 \in B_t\}.
        \end{equation*}
        \item Else, $\omega_1 \sim_t B^{m-1}F$, and then, $\omega_1 \sim_{t-1} W^{m-1} B$ by \Cref{eq:B_previous,eq:F_previous}. In that case, $\nu_t(\omega_1) = 0$ and
        \begin{equation*}
            \nu_{t-1}(\omega_1) = \begin{cases}
                -1 & \text{if } m \geq 2, \\
                0 & \text{otherwise}
            \end{cases} \quad = -\bbOne\{v_1 \in B_t\}.
        \end{equation*}
    \end{itemize}
    In both cases, we have that
    \begin{equation} \label{eq:prefix_subpath}
        \nu_t(\omega_1) = \nu_{t-1}(\omega_1) + \bbOne\{v_1 \in B_t\}.
    \end{equation}
    By a symmetric reasoning, we obtain
    \begin{equation} \label{eq:suffix_subpath}
        \nu_t(\omega_3) = \nu_{t-1}(\omega_3) - \bbOne\{v_k \in B_t\}.
    \end{equation}
    Now, we will apply the same kind of argument to some well-chosen sub-paths of~$\omega_2$, until it is entirely covered.
    Specifically, for any sub-path~$z$ of~$\omega_2$,
    \begin{alignat*}{5}
        &\text{For $k \geq 0$ \qquad} &&z \sim_t WB^kW &&\impliesplus{eq:W_previous,eq:B_previous} &&z \sim_{t-1} (W+F)\cdot W^k \cdot(W+F) &&\implies \nu_t(z) = \nu_{t-1}(z) = 0, \\
        &\text{For $k \geq 1$} &&z \sim_t FB^kW &&\impliesplus{eq:W_previous,eq:B_previous,,eq:F_previous} &&z \sim_{t-1} B W^k \cdot(W+F) &&\implies \nu_t(z) = \nu_{t-1}(z) = +1, \\
        &\text{For $k \geq 1$} &&z \sim_t WB^kF &&\impliesplus{eq:W_previous,eq:B_previous,,eq:F_previous} &&z \sim_{t-1} (W+F) \cdot W^k B &&\implies \nu_t(z) = \nu_{t-1}(z) = -1, \\
        & \text{For $k \geq 0$} &&z \sim_t FB^kF &&\impliesplus{eq:F_previous,eq:B_previous} &&z \sim_{t-1} BW^kB &&\implies \nu_t(z) = \nu_{t-1}(z) = 0, \\
        & &&z \sim_t FW &&\impliesplus{eq:F_previous,eq:WF_previous} &&z \sim_{t-1} BF &&\implies \nu_t(z) = \nu_{t-1}(z) = 0, \\
        & &&z \sim_t WF &&\impliesplus{eq:F_previous,eq:WF_previous} &&z \sim_{t-1} FB &&\implies \nu_t(z) = \nu_{t-1}(z) = 0. 
    \end{alignat*}
    Since the endpoints of~$\omega_2$ ($v_m$ and~$v_M$) are not in~$B_t$,
    any sequence of~$k$ consecutive vertices in~$\omega_2$ that are all in~$B_t$ must be sandwiched between two vertices that are not in~$B_t$.
    Therefore, any such sequence is covered by one of the first 4 equations above.
    Moreover, the last 4 equations cover all remaining edges of~$\omega_2$.
    
    Therefore, $\omega_2$ can be divided into a sequence of sub-paths~$(z_1,\ldots,z_r)$ with~$r \geq 1$, such that for every~$i \leq r$, $\nu_t(z_i) = \nu_{t-1}(z_i)$. This implies that
    \begin{equation*}
        \nu_t(\omega_2) = \sum_{i=1}^r \nu_t(z_i) = \sum_{i=1}^r \nu_{t-1}(z_i) = \nu_{t-1}(\omega_2).
    \end{equation*}
    Together with \Cref{eq:prefix_subpath,eq:suffix_subpath} and given that~$\nu_t(\omega) = \nu_t(\omega_1) + \nu_t(\omega_2) + \nu_t(\omega_3)$, this concludes the proof of \Cref{lemma:flow_conservation}.
\end{proof}



The previous result allows us to compute the flow along a path as a function of the number of beeps at its endpoints, a property sometimes referred to as {\em Ohm’s law} in similar contexts.

\begin{corollary} [Ohm's law] \label{cor:flow_path}
    For any path~$\omega$ with vertex sequence~$(v_1,\ldots,v_k)$ and every~$t \geq 0$, $\nu_t(\omega) = \beepcount_t(v_1) - \beepcount_t(v_k)$.
\end{corollary}
\begin{proof}
    Fix an arbitrary path~$\omega$ with vertex sequence~$(v_1,\ldots,v_k)$.
    We proceed by induction on~$t$.
    Since all nodes are in~$W_0$ by \Cref{eq:initialization_condition}, $\nu_0(\omega) = 0$, and~$\beepcount_0(v_1) = \beepcount_0(v_k) = 0$, so the statement holds for~$t=0$.
    Now, fix~$t > 0$ and assume that the statement holds for~$t-1$. We have
    \begin{align*}
        \nu_{t}(\omega) &= \nu_{t-1}(\omega) + \bbOne\{v_1 \in B_t\} - \bbOne\{v_k \in B_t\} & \text{(by \Cref{lemma:flow_conservation})} \\
        &= \beepcount_{t-1}(v_1) - \beepcount_{t-1}(v_k) + \bbOne\{v_1 \in B_t\} - \bbOne\{v_k \in B_t\} & \text{(by induction hypothesis)} \\
        &= \beepcount_{t}(v_1) - \beepcount_{t}(v_k),
    \end{align*}
    which concludes the induction and the proof of \Cref{cor:flow_path}.
\end{proof}


As we demonstrate in the next result, our version of Ohm's law implies that leaders with the largest number of beeps cannot be eliminated in the next round, which opens the way to establishing the correctness of our algorithm.

\begin{lemma} \label{lemma:no_leader_extinction}
    There is always at least one leader in the population.
\end{lemma}
\begin{proof}
    We start by showing that a leader cannot be eliminated unless it has beeped less than at least one neighbor.
    
    \begin{claim} \label{lemma:elimination_condition}
        If~$u \in B_t^\nlead$ with~$t \geq 1$, then there exists~$v \in N_1(u)$ such that~$\beepcount_{t-1}(u) < \beepcount_{t-1}(v)$.
    \end{claim}
    \begin{proof}
        By \Cref{eq:B_nlead_previous}, there exists~$v \in N_1(u) \cap B_{t-1}$, and by \Cref{eq:B_previous}, $u \in W_{t-1}$. Therefore, $\nu_{t-1}(u,v) = -1$.
        Moreover, by \Cref{cor:flow_path}, 
        $\nu_{t-1}(u,v) = \beepcount_{t-1}(u) - \beepcount_{t-1}(v)$, and so    
        $\beepcount_{t-1}(u) = \beepcount_{t-1}(v) - 1$, which concludes the proof of \Cref{lemma:elimination_condition}.
    \end{proof}

    Now, let~$\master_t$ be the set of nodes with the highest number of beeps up to round~$t$, and that are also leaders in round~$t$:
    \begin{equation*}
        \master_t := \argmax_{w \in V} \pa{ \beepcount_t(w) } \medcap \pa{W_t^\lead \cup B_t^\lead \cup F_t^\lead}.
    \end{equation*}
    We will show by induction that for every~$t \in \bbN$, $\master_t \neq \emptyset$.
    By \Cref{eq:initialization_condition}, no node beeps in round~$0$, so~$\master_0 = W_0^\lead \neq \emptyset$. Now, fix~$t > 0$ and assume that~$\master_{t-1} \neq \emptyset$.
    Let~$u \in \master_{t-1}$, and let~$A = \{v \in V, \quad \beepcount_{t}(v) > \beepcount_{t}(u)\}$.

    Consider the case that $A = \emptyset$.
    Note that~$u$ is still a leader in round~$t$; otherwise, since~$u$ is a leader in round~$t-1$, we would have~$u \in W^\lead_{t-1} \cap B^\nlead_t$, and \Cref{lemma:elimination_condition} gives a contradiction with the maximality of~$\beepcount_{t-1}(u)$. Therefore, $A = \emptyset$ implies that $u \in \master_t$, which concludes the induction. 

    We consider the case where~$A \neq \emptyset$.
    For every~$v \in A$, we have
    \begin{equation*}
        \beepcount_{t}(u) < \beepcount_t(v) \leq \beepcount_{t-1}(v) + 1 \leq \beepcount_{t-1}(u) + 1 \leq \beepcount_{t}(u) + 1,
    \end{equation*}
    and hence, $\beepcount_{t-1}(v) = \beepcount_{t-1}(u) = \beepcount_{t}(u)$ and $\beepcount_t(v) = \beepcount_{t-1}(v) + 1$. Therefore,
    \begin{equation} \label{eq:A_property}
        A \subseteq \argmax_{w \in V} \pa{ \beepcount_{t-1}(w) } \medcap B_t.
    \end{equation}
    Let~$v \in A$. If~$v$ is not a leader in round~$t$, then \Cref{eq:A_property} implies~$v \in B^\nlead_t$; by \Cref{lemma:elimination_condition}, this implies that~$v \notin \argmax_{w \in V} \pa{ \beepcount_{t-1}(w) }$, which contradicts \Cref{eq:A_property}. Therefore, $v$ has to be a leader in round~$t$. This implies that~$A \subseteq \master_t \neq \emptyset$, which concludes the induction and the proof of \Cref{lemma:no_leader_extinction}.
\end{proof}

To conclude this section, we prove two additional results that will be useful when computing upper bounds on the convergence time of our protocol. \Cref{lemma:max_beepcount_difference} is a straightforward consequence of Ohm's law and the maximum flow on a path, while \Cref{lemma:traveling_beep} establishes the propagation of beeps across the network.

\begin{lemma} \label{lemma:max_beepcount_difference}
    For every~$u,v \in V$, $|\beepcount_t(u) - \beepcount_t(v)| \leq \dis(u,v)$.
\end{lemma}
\begin{proof}
    Let~$\omega$ be a shortest path of length~$\dis(u,v)$ between~$u$ and~$v$.
    By \Cref{eq:max_flow}, $|\nu_t(\omega)| \leq \dis(u,v)$, and by \Cref{cor:flow_path}, $\nu_t(\omega) = \beepcount_t(u) - \beepcount_t(v)$, which concludes the proof of \Cref{lemma:max_beepcount_difference}.
\end{proof}

\begin{lemma} \label{lemma:traveling_beep}
    If~$\beepcount_t(u) > \beepcount_t(v)$, then there exists a round~$s \leq t+\dis(u,v)$ such that~$v \in B_s$.
\end{lemma}
\begin{proof}
    We begin by showing (by induction on~$k$) that the following property holds for every~$k \geq 1$:
    \begin{multline*}
        P_k := ``\forall \omega \text{ with vertex sequence } (v_0,\ldots,v_k), \forall t \in \bbN, \quad
        \nu_t(\omega) > 0 \implies \exists s \leq t+k, ~ v_k \in B_s. "
    \end{multline*}
    Let~$t \in \bbN$. If~$\omega$ consists in a single edge~$(v_0,v_1)$ and~$\nu_t(\omega) > 0$, then by \Cref{def:flow}, $v_0 \in B_t$ and~$v_1 \in W_t$. By~\Cref{eq:WB_next}, this implies that~$v_1 \in B_{t+1}$, and hence~$P_1$ holds.

    Now assume that~$P_k$ holds for some given~$k\geq 1$.
    Let~$t \in \bbN$ and~$\omega = (v_0,\ldots,v_{k+1})$ s.t.~$\nu_t(\omega) > 0$.
    If~$v_{k+1} \in B_{t+1}$, then the property holds, so we will restrict attention to the case that~$v_{k+1} \notin B_{t+1}$.
    Let
    \begin{equation*}
        i := \max \{j \leq k, \quad \nu_t(v_j,\ldots,v_{k+1}) > 0 \}.
    \end{equation*}
    Since~$\nu_t(\omega) > 0$, $i$ is well-defined. Since the flow along the path may only vary by 1 unit with each edge, and by the maximality of~$i$, we have
    \begin{equation*}
        \nu_t\pa{v_{i+1},\ldots,v_{k+1}} = 0,
    \end{equation*}
    and moreover,~$v_i \in B_t$ and~$v_{i+1} \in W_t$. By \Cref{eq:WB_next}, this implies~$v_{i+1} \in B_{t+1}$. Therefore, since~$v_{k+1} \notin B_{t+1}$ and by \Cref{lemma:flow_conservation},
    \begin{equation*}
        \nu_{t+1}(v_{i+1},\ldots,v_{k+1}) = \nu_t(v_{i+1},\ldots,v_{k+1}) + 1 = 1,
    \end{equation*}
    and we can conclude that~$P_{k+1}$ holds, by applying~$P_k$ to~$(v_{i+1},\ldots,v_{k+1})$ in round~$t+1$. By induction, this implies that~$P_k$ holds for every~$k \geq 1$.
    
    \Cref{lemma:traveling_beep} follows by applying~$P_k$ to a shortest path~$\omega$ of length~$\dis(u,v)$ between~$u$ and~$v$, given that~$\nu_t(\omega) = \beepcount_t(u) - \beepcount_t(v) > 0$.
\end{proof}

\section{Probabilistic Analysis} \label{sec:probabilistic_analysis}

\subsection{A Preliminary Result}

The proof of the following theorem uses classical arguments, and is deferred to \Cref{appendix:proba}.


\begin{theorem} \label{thm:visits_anti_concentration} Let~$(X_t)_{t \geq 1}$ be an aperiodic, irreducible Markov chain with finite state space~$\mathcal{X}$ and admitting a stationary distribution $\pi$. For~$x \in \mathcal{X}$ and~$t \in \bbN$, let~$N_t(x)$ be the number of visits to state~$x$ up to round~$t$ (included):
    \begin{equation*}
        N_t(x) := \sum_{s=1}^t \bbOne \{ X_s = x \}.
    \end{equation*}
    If~$|\mathcal{X}|>1$ and~$X_1 \sim \pi$, then for every~$x \in \mathcal{X}$, and every~$t$ large enough, and every constant $c>0$, there exists $\varepsilon = \varepsilon(c)$ such that $0<\varepsilon <1$ and
    \begin{equation*}
       \sup_{m \in \mathbb{N}} \bbP \pa{|N_t(x) - m| \leq c\sqrt{\Var(N_t(x))}}\leq 1-\varepsilon.
    \end{equation*}
\end{theorem}

%

\subsection{\texorpdfstring{Proof of \Cref{thm:main}}{Proof of Theorem 2}} \label{sec:main_proof}

In this section, we consider Markov chains on the state space~$\mathcal{X} := \{W,B,F\}$ with transition matrix~$P$:
\begin{equation} \label{eq:transition_matrix}
\vcenter{\hbox{\includegraphics[width=4cm,height=4cm]{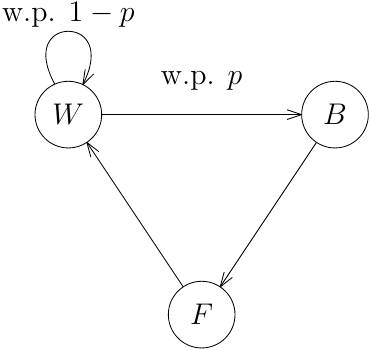}}}
\qquad\qquad
P := 
    \begin{pmatrix}
    1-p & p & 0\\
    0 & 0 & 1\\
    1 & 0 & 0
    \end{pmatrix}
\end{equation}


More precisely, we consider the product~$(\mathbf{X}_t)_{t \in \mathbb{N}}$ of~$n$ i.i.d. Markov chains~$\{(X^{(u)}_t)_{t\in \bbN}, u \in V\}$ with transition matrix~$P$, and such that~$X^{(u)}_1  = W$ for every~$u \in V$.
It is easy to verify that each $(X_t^{(u)})_{t \in \bbN}$ has stationary distribution 
\begin{equation}
    \pi := (\pi_W, \pi_B, \pi_F) = \pa{\tfrac{1}{2p+1},\tfrac{p}{2p+1},\tfrac{p}{2p+1}}.
    \label{eq:stationary_distribution}
\end{equation}

We write~$N^{(u)}_t$ to denote the number of visits to state~$B$ of~$(X^{(u)}_t)_{t\in \bbN}$ up to round~$t$ (included).
Finally, we define
\begin{equation}
    \sigma_{u,v} = \inf \left\{ t \in \bbN, \quad |N^{(u)}_t - N^{(v)}_t| > D \right\},
    \label{eq:def_sigma}
\end{equation}
where~$D$ (the diameter of the graph) may depend on~$n$.


The following lemma gives an anti-concentration result for the random variables $N_t^{(u)}$, derived from \Cref{thm:visits_anti_concentration}.


\begin{lemma} \label{lem:anti_concentration_Nt}
    There exists a constant $\varepsilon = \varepsilon(p)>0$ such that, for any $u \in V$ and for any large enough~$t$,
    \begin{equation*}
        \sup_{m \in \mathbb{N}}\Prob{|N_t^{(u)}-m| \leq \sqrt{t}} \leq 1-\varepsilon.
    \end{equation*}
\end{lemma}

\begin{proof}
    Fix~$u \in V$, and let $\tau$ be a random variable representing the first return of~$X_t^{(u)}$ to state~$B$, when starting in state~$B$.
    From the definition of the Markov Chain $(X_t^{(u)})_t$, we have that $\tau \sim 2 + \mathrm{Geom}(p)$. Then, let $\tau_1, \tau_2,\dots$ be i.i.d. copies of $\tau$. We have that
    \begin{equation} \label{eq:alternative_def}
        N_t^{(u)} = \min\{ k \geq 0 : \tau_1 + \dots + \tau_{k+1} > t\}.
    \end{equation}
    Therefore, for the tail-sum formula for the expectation and any large enough $t$, since $\Expc{N_t^{(u)}} = \pi_B \, t = \frac{p}{2p+1}t$ from \cref{eq:stationary_distribution}, it holds
\begin{align*}
    \Expc{|N_t^{(u)}-\Expc{N_t^{(u)}}|} &= \sum_{k=1}^{\infty} \Prob{|N_t^{(u)}-\Expc{N_t^{(u)}}| \geq k}
    \\ & \geq \sum_{k=1}^{\infty} \Prob{N_t^{(u)} \geq \frac{p}{2p+1}t + k} 
    \\ & = \sum_{k=1}^{\infty} \Prob{\tau_1+ \dots + \tau_{pt/(2p+1) + k} \leq t}. 
    & \text{(by \Cref{eq:alternative_def})}
\end{align*}
We can write~$\tau_i = 2+W_i$ where~$W_1,W_2, \dots$ are i.i.d. $\mathrm{Geom}(p)$ random variables. With this notation,
\begin{equation*}
    \sum_{i=1}^{\frac{pt}{2p+1}+k} \tau_i = \frac{2pt}{2p+1}+2k+ \sum_{i=1}^{\frac{pt}{2p+1}+k} W_i,
\end{equation*}
and hence,
\begin{equation*}
    \Prob{\sum_{i=1}^{\frac{pt}{2p+1}+k} \tau_i \leq t} = \Prob{\sum_{i=1}^{\frac{pt}{2p+1}+k}W_i \leq t- 2\frac{pt}{2p+1} - 2k}.
\end{equation*}
It holds that, from \cref{lem:sum_geometric}, if $Z = \mathrm{Bin}( t- 2\frac{pt}{2p+1} - 2k, p)$, then there exists a constant $\delta = \delta(p)>0$ such that 
\begin{align*}
    \Prob{\sum_{i=1}^{\frac{pt}{2p+1}+k}W_i \leq t- 2\frac{pt}{2p+1} - 2k} &= \Prob{Z \geq \frac{pt}{2p+1}+k}
    \\ & = \Prob{Z \geq \Expc{Z} + k(2p+1)} \\ & \geq \frac{1}{2},
\end{align*}
for any $k \leq \delta \sqrt{t}$, from a standard application  of the Berry-Esseen
theorem (see for instance the proof of \cref{thm:visits_anti_concentration}).
Hence, we have that
\[\Expc{|N_t^{(u)} - \Expc{N_t^{(u)}}|} \geq \frac{\delta}{2} \sqrt{t} \]
and hence, from Jensen's Inequality
\[\mathrm{Var}(N_t^{(u)}) \geq \left(\Expc{|N_t^{(u)} - \Expc{N_t^{(u)}}|}\right)^2 \geq \frac{\delta^2}{4} t.\]

Applying \cref{thm:visits_anti_concentration} with $c = 2/\delta$ and  from the previous inequality, we have that there exists $\varepsilon = \varepsilon(p)$ such that
\begin{equation*}
\sup_{m \in \mathbb{N}} \Prob{|N_t^{(u)}- m| \leq \sqrt{t}} \leq 
\sup_{m \in \mathbb{N}} \Prob{|N_t^{(u)}- m| \leq \frac{2}{\delta}\sqrt{\Var(N_t^{(u)})}} \leq 1-\varepsilon,
\end{equation*}
which concludes the proof.
\end{proof}

The previous result can be used to show that for any~$d > 0$, two independent sequences~$(N^{(u)}_t)_{t \in\bbN}$ and~$(N^{(v)}_t)_{t \in\bbN}$ differ by at least~$d$ after only~$d^2$ rounds with constant probability.

\begin{lemma} \label{lemma:one_step_deviation}
There exists a constant $\varepsilon=\varepsilon(p)>0$ such that, for every initial configuration $\mathbf{X}_0 \in \mathcal{X}^{n}$, any pair of nodes $u,v \in V$ and any large enough $d$,
    \begin{equation*}
        \bbP_{\mathbf{X}_0} \pa{ |N^{(u)}_{d^2} - N^{(v)}_{d^2}| <  d} \leq 1-\varepsilon.
    \end{equation*}
\end{lemma}
\begin{proof}
    Let~$(\Tilde{X}_t)_{t \geq t_0}$ be a Markov chain with the same transition matrix $P$, and such that it starts from the stationary distribution, i.e. $\Tilde{X}_0 \sim \pi$. 
    We define a coupling between $(\Tilde{X}_t)_{t \in \mathbb{N}}$ and $X_t^{(u)}$ such that $\Tilde{X}_t$ and $X_t^{(u)}$ move around independently until they collide, and after this time they stick together and move together. In other word, if we define
    \[\tau = \min\{t \geq 0: \Tilde{X}_t = X_t^{(u)}\},\]
    then we have that
    \[\Tilde{X}_t = X_t^{(u)} \quad \text{for every $t \geq \tau$.}\]
    Denote with $\Tilde{\bbP}_{\mathbf{X}_0}(\cdot)$ the joint distribution of $(\Tilde{X}_t, X_t^{(u)})_t$, with $X_t^{(u)}$ starting from $\mathbf{X}_0(u)$.

    
    \begin{claim} \label{claim:cant_outrun}
        For every~$t \geq 0$,  $\Tilde{\bbP}_{\mathbf{X}_0}\pa{ |\Tilde{N}_t - N_t^{(u)}|\leq 1}= 1$
    \end{claim}
    \begin{proof}
    Let $T$ be the first meeting time of $\Tilde{N}_t$ and $N_t^{(u)}$. Assume that there exists $t < T$ such that $|\Tilde{N}_t - N_t^{(u)}| \geq 2$.
    W.l.g., assume that $\Tilde{N}_t \geq 2+ N_t^{(u)}$. This implies that the walk $\Tilde{X}_t$ did at least one more complete clockwise turn of the cycle in the picture in \cref{eq:transition_matrix} respect to $X_t^{(u)}$. But this contradicts the fact that $t < T$, since in that case $\Tilde{X}_t$ and $X_t^{(u)}$ necessarily met before or during time $t$. Then, for every round $t \geq T$ it holds $X_t^{(u)} = \Tilde{X}_t$, so this proves the lemma.
      \end{proof}
        Then, taking $\varepsilon>0$ as in \cref{lem:anti_concentration_Nt}, we have that
        \begin{align*}
            & \bbP_{\mathbf{X}_0}\pa{|N^{(u)}_{d^2}-N^{(v)}_{d^2}| <  d} 
            \\ & = \Tilde{\bbP}_{\mathbf{X}_0} \pa{|(N_{d^2}^{(u)}-\Tilde{N}_{d^2}) + (\Tilde{N}_{d^2} - N_{d^2}^{(v)})| <  d}  & \text{(by the coupling definition)}
            \\ & \leq \Tilde{\bbP}_{\mathbf{X}_0} \pa{|N_{d^2}^{(u)}- \Tilde{N}_{d^2}| + |\Tilde{N}_{d^2}- N_{d^2}^{(v)}| <  d} 
            \\ & \leq \Tilde{\bbP}_{\mathbf{X}_0} \pa{|\Tilde{N}_{d^2} - N_{d^2}^{(v)}|\leq  d} &\text{(by \cref{claim:cant_outrun})}
            \\ & = \sum_{m \in \mathbb{N}}\Tilde{\bbP}_{\mathbf{X}_0}\pa{|\Tilde{N}_{d^2}- m| \leq  d \mid N_{d^2}^{(v)}= m} \bbP_{\mathbf{X}_0}\pa{N_{d^2}^{(v)} = m}
            \\ & = \sum_{m \in \mathbb{N}} \bbP \pa{|\Tilde{N}_{d^2}-m| \leq d} \bbP_{\mathbf{X}_0} \pa{N_{d^2}^{(v)} = m} & \text{(independence of $\Tilde{N}_{d^2}$ and $N_{d^2}^{(v)}$)}
            \\ &\leq (1-\varepsilon)\sum_{m \in \mathbb{N}}\bbP_{\mathbf{X}_0}\pa{N_{d^2}^{(v)} = m} & \text{(by \cref{lem:anti_concentration_Nt})}
            \\ & \leq 1-\varepsilon.
        \end{align*}
\end{proof}

In the following lemma, we show an upper bound on the variables~$\sigma_{u,v}$ that holds with high probability. This is achieved by using the results obtained so far, and adding a~$\log n$ factor.

\begin{lemma} \label{lemma:upper_bound_tau}
    There exists a constant~$A > 0$ s.t., for any $u,v \in V$
    \begin{equation*}
        \bbP(\sigma_{u,v} > A D^2 \log n) \leq n^{-4}.
    \end{equation*}
\end{lemma}
\begin{proof}
   Let~$\mathcal{E}_k$ the event~$\{|N^{(u)}_{k D^2} - N^{(v)}_{k D^2}| \leq D\}$.
    By definition of $\sigma_{u,v}$ in \cref{eq:def_sigma}, for every~$A > 0$,
    \begin{equation*}
        \{ \sigma_{u,v} > A D^2 \log n \} \implies \bigcap_{k=1}^{\lfloor A \log n \rfloor} \mathcal{E}_k.
    \end{equation*}
    Therefore, by the chain rule, \Cref{lemma:one_step_deviation} and the Markov property 
    \begin{equation*}
        \bbP(\sigma_{u,v} > A D^2 \log n) \leq \bbP \pa{ \mathcal{E}_1 } \cdot \prod_{k=2}^{\lfloor A \log n \rfloor} \bbP \pa{ \mathcal{E}_k \mid \mathcal{E}_{k-1} } \leq (1-\varepsilon)^{\lfloor A \log n \rfloor} \leq \frac{1}{n^2},
    \end{equation*}
    where~$\varepsilon = \varepsilon(p)>0$ as in \cref{lemma:one_step_deviation}, and we consider $A \geq 4 \log^{-1}\pa{\frac{1}{1-\varepsilon}}$, which concludes the proof.
\end{proof}


Finally, we combine \Cref{lemma:upper_bound_tau} with insights on the flow from \Cref{sec:flow} to show that, for any pair of leaders in the graph, the difference in their numbers of beeps becomes large enough that paths between the two nodes are eventually saturated by beep waves in one direction. This necessarily leads to the elimination of the leader with fewer beeps, and \Cref{thm:main} follows from a union bound.

\begin{proof} [Proof of \Cref{thm:main}]
    For every node~$u$, as long as~$u$ is a leader and by definition of the protocol, we can couple the state of~$u$ with~$(X_t^{(u)})_{t \in \bbN}$. In other words, up to the round where~$u$ is eliminated,
    \begin{align*}
        u \in W^\lead_t &\implies X_t^{(u)} = W, \\
        u \in B^\lead_t &\implies X_t^{(u)} = B, \\
        u \in F^\lead_t &\implies X_t^{(u)} = F.
    \end{align*}
    If~$u$ is eliminated, we simply stop enforcing the coupling in subsequent rounds.
    \begin{claim} \label{claim:elimination_condition}
        In round~$\sigma_{u,v}$, either~$u$ or~$v$ has been eliminated.
    \end{claim}
    \begin{proof}
        Assume for the sake of contradiction that both~$u$ and~$v$ are still leaders in round~$t = \sigma_{u,v}$.
        Then by construction, $N_t^{(u)} = \beepcount_t(u)$ and $N_t^{(v)} = \beepcount_t(v)$.
        Moreover, by definition of~$\sigma_{u,v}$, $|N_t^{(u)}-N_t^{(v)}| > D$.
        This implies~$|\beepcount_t(u) - \beepcount_t(v)| > \dis(u,v)$, which is in contradiction with \Cref{lemma:max_beepcount_difference}.
        This concludes the proof of \Cref{claim:elimination_condition}.
    \end{proof}
    By letting~$A$ be the constant in \Cref{lemma:upper_bound_tau}, we have that
    \begin{align*}
        &\bbP ( \text{There are still 2 leaders in round } AD^2 \log n ) \\
        &\leq \bbP \pa{ \bigcup_{u,v \in V} \{\sigma_{u,v} > AD^2 \log n \} } & \text{(by \Cref{claim:elimination_condition})} \\
        &\leq n^2 \cdot n^{-4} = n^{-2}, & \text{(by \Cref{lemma:upper_bound_tau} and union bound)}
    \end{align*}
    which concludes the proof of \Cref{thm:main}.
\end{proof}

\subsection{\texorpdfstring{Proof of \Cref{thm:side}}{Proof of Theorem 3}} \label{sec:side_result}

In this section, we restrict attention to the case that~$p = 1/(D+1)$. We use the same approach as for the proof of \Cref{thm:main}.
In this regime, every~$\Theta(D)$ rounds and for any pair of leaders~$(u,v)$, there is a constant probability that~$u$ beeps and not~$v$.
Insights from \Cref{sec:flow} imply that in that case, $v$ is eliminated.
Therefore, we show that convergence happens in only~$O(D \log n)$ rounds with high probability, and without the need for anti-concentration results.

\begin{lemma} \label{lemma:elimination}
    Assume that~$p = 1/(D+1)$, and fix a round~$t \in \bbN$.
    If~$u$ and~$v$ are both leaders in round~$t$, then with constant probability, at least one of them is eliminated before round~$t+2D+1$.
\end{lemma}
\begin{proof}
    First, note that since~$D \geq 1$,
    \begin{equation} \label{eq:simple_bounds}
        \frac{1}{4} \leq \pa{1-\frac{1}{D+1}}^{D+1} < \frac{1}{e},
    \end{equation}
    which follows by analyzing the function~$x \rightarrow (1-1/x)^x$.
    In what follows, 
    we consider the same coupling as in the proof of \Cref{thm:main}, between the chain~$(X_t^{(u)})_{t \in \bbN}$, and the state of~$u$ up to the round where it is eliminated.
    We write~$N_{t_1:t_2}^{(u)}$ to denote the number of visits of~$(X_t^{(u)})_{t \in \bbN}$ to state~$B$ between rounds~$t_1$ and~$t_2$.

    First, we consider the case that~$|\beepcount_t(u) - \beepcount_t(v)| \neq 0$. 
    Without loss of generality, we assume that~$\beepcount_t(u) > \beepcount_t(v)$. Let~$A$ be the event that~$v$ does not beep (by itself) between rounds~$t$ and~$t+2(D+1)$, that is,
    \begin{equation*}
        A := \left\{ N_{t:t+2(D+1)}^{(v)} = 0 \right\}.
    \end{equation*}
    By \Cref{eq:simple_bounds},
    \begin{equation*}
        \bbP(A) = \pa{1-p}^{2(D+1)} = \pa{1-\frac{1}{D+1}}^{2(D+1)} \geq \frac{1}{4^2} = \frac{1}{16},
    \end{equation*}
    i.e., $A$ happens with constant probability.
    Moreover, by \Cref{lemma:traveling_beep}, $v$ must beep at least once between rounds~$t$ and~$t+D$. Therefore, conditioning on~$A$, $v$ beeps in a round~$t' \leq t+D$ where~$X_{t'}^{(v)} \neq B$; this implies that~$v \in B^\nlead_{t'}$, i.e., $v$ is not a leader in round~$t'$.
    
    Now, consider the case that~$\beepcount_t(u) = \beepcount_t(v)$. Let~$A$ be defined as before, and~$B$ the event that~$u$ beeps at least once (by itself) between rounds~$t$ and~$t+(D+1)$, that is,
    \begin{equation*}
        B := \left\{ N_{t:t+(D+1)}^{(u)} \geq 1 \right\}.
    \end{equation*}
    By \Cref{eq:simple_bounds},
    \begin{equation*}
        \bbP(B) = 1-\pa{1-p}^{(D+1)} \geq 1-\frac{1}{e}.
    \end{equation*}
    Since~$(X_t^{(u)})_{t \in \bbN}$ and~$(X_t^{(v)})_{t \in \bbN}$ are independent, $A \cap B$ happens with constant probability.
    Conditioning on~$A \cap B$, there is a round~$s \leq t+(D+1)$ such that~$\beepcount_s(u) \geq 1$, and we can conclude as before with \Cref{lemma:traveling_beep} that in some round~$t' \leq s+D \leq t+2D+1$, $v$ is not a leader anymore.
    This establishes \Cref{lemma:elimination}.
\end{proof}

\begin{proof} [Proof of \Cref{thm:side}]
    For~$u,v \in V$, let~$A_t^{(u,v)}$ the event that~$u$ and~$v$ are both leaders in round~$t$. By \Cref{lemma:elimination}, for every~$t \in \bbN$, there is a constant~$c < 1$ s.t.
    \begin{equation*}
        \bbP \pa{A_{t+2D+1}^{(u,v)} \mid A_t^{(u,v)}} \leq c.
    \end{equation*}
    By taking~$k \geq 4/\log \frac{1}{c}$, we obtain
    \begin{equation*}
        \bbP\pa{A_{k \log n \cdot (2D+1)}^{(u,v)}} = \bbP \pa{A_0^{(u,v)}} \prod_{i=0}^{k \log n-1} \bbP \pa{A_{(i+1)(2D+1)}^{(u,v)} \mid A_{i(2D+1)}^{(u,v)}} \leq c^{k \log n} \leq n^{-4},
    \end{equation*}
    and we can conclude (using an union bound over all pairs of agents) that a single leader remains in round~$k \log n \cdot (2D+1)$ with high probability, which concludes the proof of \Cref{thm:side}.
\end{proof}

\section{Discussion}

In this article, we identify and analyze a simple protocol for solving Leader Election within the beeping model.
Beyond its natural simplicity, our algorithm is remarkably resource-efficient: it requires only six memory states, relies solely on fair coin tosses, and does not depend on unique identifiers or any prior knowledge of the underlying communication graph. This distinguishes it from most existing protocols for this problem, where the use of unique identifiers is standard, which implies a memory space dependent on~$n$.
The main drawbacks of our approach are the overhead in convergence time, which scales approximately linearly with the graph’s diameter, and no termination detection.

Beyond directly comparing performance and resource efficiency, we believe our contribution differs slightly in nature from most works in the literature. While they usually focus on minimizing computation time and obtaining desirable properties, we are interested in exploring what can be computed by elementary protocols -- those likely to emerge spontaneously in simple organisms through evolution -- when communication is restricted to minimal signals, such as beeps. In that regard, our approach is similar to the one of~\cite{gilbert_computational_2015}. We think both perspectives contribute in their own way to a deeper understanding of the capabilities of distributed systems.

The strongest obstacle to the biological plausibility of our algorithm lies in the assumption that the initial configuration contains at least one leader, and no agent in a beeping state. Relaxing this assumption without compromising the protocol’s simplicity appears challenging. Indeed, if the initial configuration was arbitrary, it could include persistent and deterministic beep waves traveling along cycles of the graph, while no leader would be present in the network. These deterministic waves could be almost indistinguishable, from the point of view of each node, from those that would be randomly emitted by a leader in a correct configuration.
Identifying another simple but more robust rule remains an open question for future research.

Finally, although we do not demonstrate it here, we believe that the upper bound in Theorem 1 is tight, up to a factor of~$\log n$. To illustrate this, consider a configuration in which the only two leaders are positioned at the ends of a path of length~$D$. The point where the waves emitted by each leader meet appears to move over time like a simple random walk. The time it takes for one of the leaders to be eliminated corresponds to the time it takes for this random walk to reach one of the ends, which is~$\Theta(D^2)$. If this argument proves correct, it would suggest the existence of a space-time trade-off in this model, which would be yet another intriguing research direction.

\paragraph{Acknowledgment}
The authors wish to thank Francesco D'Amore for engaging discussions on the topic.
Part of the work was done while the authors were affiliated with Bocconi University, Milan.

\bibliographystyle{plain}
\bibliography{references}

\begin{thebibliography}{10}

\bibitem{AlistarhAEGR17}
Dan Alistarh, James Aspnes, David Eisenstat, Rati Gelashvili, and Ronald~L. Rivest.
\newblock Time-space trade-offs in population protocols.
\newblock In Philip~N. Klein, editor, {\em Proceedings of the Twenty-Eighth Annual {ACM-SIAM} Symposium on Discrete Algorithms, {SODA} 2017, Barcelona, Spain, Hotel Porta Fira, January 16-19}, pages 2560--2579. {SIAM}, 2017.

\bibitem{alistarh_nearoptimal_2022}
Dan Alistarh, Joel Rybicki, and Sasha Voitovych.
\newblock Near-optimal leader election in population protocols on graphs.
\newblock In {\em Proceedings of the 2022 {{ACM Symposium}} on {{Principles}} of {{Distributed Computing}}}, {{PODC}}'22, pages 246--256, New York, NY, USA, July 2022. Association for Computing Machinery.

\bibitem{AspnesR09}
James Aspnes and Eric Ruppert.
\newblock An introduction to population protocols.
\newblock In Beno{\^{\i}}t Garbinato, Hugo Miranda, and Lu{\'{\i}}s E.~T. Rodrigues, editors, {\em Middleware for Network Eccentric and Mobile Applications}, pages 97--120. Springer, 2009.

\bibitem{BerenbrinkGK20}
Petra Berenbrink, George Giakkoupis, and Peter Kling.
\newblock Optimal time and space leader election in population protocols.
\newblock In Konstantin Makarychev, Yury Makarychev, Madhur Tulsiani, Gautam Kamath, and Julia Chuzhoy, editors, {\em Proceedings of the 52nd Annual {ACM} {SIGACT} Symposium on Theory of Computing, {STOC} 2020, Chicago, IL, USA, June 22-26, 2020}, pages 119--129. {ACM}, 2020.

\bibitem{billingsley_probability_1995}
Patrick Billingsley.
\newblock {\em Probability and Measure}.
\newblock Wiley Series in Probability and Mathematical Statistics. Wiley, New York, 3rd ed edition, 1995.

\bibitem{chlamtac_broadcasting_1985}
I.~Chlamtac and S.~Kutten.
\newblock On {{Broadcasting}} in {{Radio Networks}} - {{Problem Analysis}} and {{Protocol Design}}.
\newblock {\em IEEE Transactions on Communications}, 33(12):1240--1246, December 1985.

\bibitem{chlebus_electing_2012}
Bogdan~S. Chlebus, Dariusz~R. Kowalski, and Andrzej Pelc.
\newblock Electing a leader in multi-hop radio networks.
\newblock In Roberto Baldoni, Paola Flocchini, and Ravindran Binoy, editors, {\em Principles of {{Distributed Systems}}}, pages 106--120, Berlin, Heidelberg, 2012. Springer.

\bibitem{CornejoK10}
Alejandro Cornejo and Fabian Kuhn.
\newblock Deploying wireless networks with beeps.
\newblock In Nancy~A. Lynch and Alexander~A. Shvartsman, editors, {\em Distributed Computing, 24th International Symposium, {DISC} 2010, Cambridge, MA, USA, September 13-15, 2010. Proceedings}, volume 6343 of {\em Lecture Notes in Computer Science}, pages 148--162. Springer, 2010.

\bibitem{CzumajD19}
Artur Czumaj and Peter Davies.
\newblock Leader election in multi-hop radio networks.
\newblock {\em Theor. Comput. Sci.}, 792:2--11, 2019.

\bibitem{DotyS18}
David Doty and David Soloveichik.
\newblock Stable leader election in population protocols requires linear time.
\newblock {\em Distributed Comput.}, 31(4):257--271, 2018.

\bibitem{DufoulonBB18}
Fabien Dufoulon, Janna Burman, and Joffroy Beauquier.
\newblock Beeping a deterministic time-optimal leader election.
\newblock In Ulrich Schmid and Josef Widder, editors, {\em 32nd International Symposium on Distributed Computing, {DISC} 2018, New Orleans, LA, USA, October 15-19, 2018}, volume 121 of {\em LIPIcs}, pages 20:1--20:17. Schloss Dagstuhl - Leibniz-Zentrum f{\"{u}}r Informatik, 2018.

\bibitem{EmekK21}
Yuval Emek and Eyal Keren.
\newblock A thin self-stabilizing asynchronous unison algorithm with applications to fault tolerant biological networks.
\newblock In Avery Miller, Keren Censor{-}Hillel, and Janne~H. Korhonen, editors, {\em {PODC} '21: {ACM} Symposium on Principles of Distributed Computing, Virtual Event, Italy, July 26-30, 2021}, pages 93--102. {ACM}, 2021.

\bibitem{emek_stone_2013}
Yuval Emek and Roger Wattenhofer.
\newblock Stone age distributed computing.
\newblock In {\em Proceedings of the 2013 {{ACM}} Symposium on {{Principles}} of Distributed Computing}, {{PODC}} '13, pages 137--146, New York, NY, USA, July 2013. Association for Computing Machinery.

\bibitem{ForsterSW14}
Klaus{-}Tycho F{\"{o}}rster, Jochen Seidel, and Roger Wattenhofer.
\newblock Deterministic leader election in multi-hop beeping networks - (extended abstract).
\newblock In Fabian Kuhn, editor, {\em Distributed Computing - 28th International Symposium, {DISC} 2014, Austin, TX, USA, October 12-15, 2014. Proceedings}, volume 8784 of {\em Lecture Notes in Computer Science}, pages 212--226. Springer, 2014.

\bibitem{GhaffariH13}
Mohsen Ghaffari and Bernhard Haeupler.
\newblock Near optimal leader election in multi-hop radio networks.
\newblock In Sanjeev Khanna, editor, {\em Proceedings of the Twenty-Fourth Annual {ACM-SIAM} Symposium on Discrete Algorithms, {SODA} 2013, New Orleans, Louisiana, USA, January 6-8, 2013}, pages 748--766. {SIAM}, 2013.

\bibitem{giakkoupis_distributed_2023}
George Giakkoupis and Isabella Ziccardi.
\newblock Distributed self-stabilizing mis with few states and weak communication.
\newblock In {\em Proceedings of the 2023 {ACM} {Symposium} on {Principles} of {Distributed} {Computing}}, {PODC} '23, pages 310--320, New York, NY, USA, June 2023. Association for Computing Machinery.

\bibitem{gilbert_computational_2015}
Seth Gilbert and Calvin Newport.
\newblock The computational power of beeps.
\newblock In Yoram Moses, editor, {\em Distributed {Computing}}, pages 31--46, Berlin, Heidelberg, 2015. Springer.

\bibitem{itai_symmetry_1990}
Alon Itai and Michael Rodeh.
\newblock Symmetry breaking in distributed networks.
\newblock {\em Information and Computation}, 88(1):60--87, September 1990.

\bibitem{kowalski_leader_2009}
Dariusz~R. Kowalski and Andrzej Pelc.
\newblock Leader election in ad hoc radio networks: A keen ear helps.
\newblock In Susanne Albers, Alberto {Marchetti-Spaccamela}, Yossi Matias, Sotiris Nikoletseas, and Wolfgang Thomas, editors, {\em Automata, {{Languages}} and {{Programming}}}, pages 521--533, Berlin, Heidelberg, 2009. Springer.

\bibitem{SudoM20}
Yuichi Sudo and Toshimitsu Masuzawa.
\newblock Leader election requires logarithmic time in population protocols.
\newblock {\em Parallel Process. Lett.}, 30(1):2050005:1--2050005:13, 2020.

\bibitem{tikhomirov_convergence_1981}
AN~Tikhomirov.
\newblock On the convergence rate in the central limit theorem for weakly dependent random variables.
\newblock {\em Theory of Probability and its Applications}, 25(4):790, 1981.

\end{thebibliography}

\clearpage
\appendix

\section{Probabilistic Tools} \label{appendix:proba}


We write~$\Phi$ to denote the density function of the normal law~$\mathcal{N}(0,1)$:
\begin{equation} \label{eq:phi}
    \Phi(z) := \frac{1}{\sqrt{2\pi}} \int_{-\infty}^z e^{-t^2/2} dt.
\end{equation}

Let~$(X_n)_{n \in \bbN}$ be a sequence of random variables on a given probability space~$(\Omega,\mathcal{F},\bbP)$. $(X_n)_{n \in \bbN}$ is said to be {\em stationary} if its joint probability distribution is invariant over time: for every~$n_1,n_2,k \in \bbN$, $(X_{n_1},\ldots,X_{n_1+k})$ and~$(X_{n_2},\ldots,X_{n_2+k})$ are equal in distribution. 

For any two $\sigma$-algebra~$\mathcal{A}$ and~$\mathcal{B} \subseteq \mathcal{F}$, let
\begin{equation*}
    \alpha(\mathcal{A},\mathcal{B}) := \sup_{A \in \mathcal{A}, B \in \mathcal{B}} |\bbP(A\cap B) - \bbP(A) \, \bbP(B)|.
\end{equation*}
For~$i,j \in \bbN \cup \{+\infty\}$, let~$\mathcal{F}_i^j$ be the~$\sigma$-algebra generated by the random variables~$\{X_k, i \leq k \leq j\}$.
$(X_n)_{n \in \bbN}$ is said to be {\em $\alpha$-mixing} if
\begin{equation*}
    \alpha(n) := \sup_{k \in \bbN} \, \alpha \pa{ \mathcal{F}_{0}^{k} , \mathcal{F}_{k+n}^{+\infty}} \longrightarrow 0.
\end{equation*}

The following result can be found in~\cite[Theorem 2 with~$\delta=1$]{tikhomirov_convergence_1981}.
\begin{theorem} [Berry-Esseen Theorem for Strongly Mixing Sequences] \label{thm:alpha_mixing_berry_esseen}
    Suppose that~$(X_n)_{n \in \bbN}$ is stationary and~$\alpha$-mixing, with~$\bbE(X_n) = 0$ and~$\bbE(|X_n|^3) < +\infty$.
    Let~$s_n := \sum_{k=1}^n X_k$, and
    \begin{equation*}
        S_n := \frac{s_n}{\sqrt{\Var(s_n)}} = \frac{s_n}{\sqrt{\bbE\left[\pa{\sum_{k=1}^n X_k }^2\right]}}.
    \end{equation*}
    If there are positive constants~$K,\beta$ such that
    \begin{equation*}
        \alpha(n) \leq K e^{-\beta n},
    \end{equation*}
    then there is a constant~$A := A(K,\beta)$ such that
    \begin{equation*}
        \sup_{z \in \bbR} |\bbP(S_n < z) - \Phi(z)| \leq A \, \frac{\log^2 n}{\sqrt{n}}.
    \end{equation*}
\end{theorem}

\begin{reptheorem}{thm:visits_anti_concentration}
    Let~$(X_t)_{t \geq 1}$ be an aperiodic, irreducible Markov chain with finite state space~$\mathcal{X}$ and admitting a stationary distribution. For~$x \in \mathcal{X}$ and~$t \in \bbN$, let~$N_t(x)$ be the number of visits to state~$x$ up to round~$t$ (included):
    \begin{equation*}
        N_t(x) := \sum_{s=1}^t \bbOne \{ X_s = x \}.
    \end{equation*}
    If~$|\mathcal{X}|>1$ and~$X_1 \sim \pi$, then for every~$x \in \mathcal{X}$, and every~$t$ large enough, and every constant $c>0$, there exists $\varepsilon = \varepsilon(c)$ such that $0<\varepsilon <1$ and
    \begin{equation*}
       \sup_{m \in \mathbb{N}} \bbP \pa{|N_t(x) - m| \leq c\sqrt{\Var(N_t(x))}}\leq 1-\varepsilon.
    \end{equation*}
\end{reptheorem}

\begin{proof}
    Let~$Y_t := \bbOne \{ X_t = x \} - \pi_x$, and for~$i,j \in \mathcal{X}$, let
    \begin{equation*}
        p_{i,j}^{(t)} := \bbP(X_{t} = j \mid X_0 = i).
    \end{equation*}
    By the convergence theorem for finite Markov chains (see, e.g., \cite[Theorem 8.9]{billingsley_probability_1995}), there are constants~$\rho \in (0,1)$ and $C>0$ such that for every~$i,j \in \mathcal{X}$,
    \begin{equation} \label{eq:convergence_theorem}
        |p_{i,j}^{(t)} - \pi_j| \leq C \rho^t.
    \end{equation}
    Therefore, by a classical argument (see, e.g., \cite[Example 27.6]{billingsley_probability_1995}), $(Y_t)_{t \in \bbN}$ is~$\alpha$-mixing with~$\alpha(t) \leq |\mathcal{X}| \, C \rho^t$.
    Moreover, $\bbE(Y_t) = 0$, $\bbE(|Y_t|^3) < +\infty$, and~$(Y_t)_{t \in \bbN}$ is stationary (since~$(X_t)_{t \in \bbN}$ is stationary), so we can apply \Cref{thm:alpha_mixing_berry_esseen}: there is a constant~$A := A(|\mathcal{X}|,C,\rho)$ such that
    \begin{equation} \label{eq:applied_berry_esseen}
        \sup_{z \in \bbR} \quad \left| ~ \bbP\pa{\frac{N_t(x) - t \, \pi_x}{\sqrt{\Var(N_t(x))}} < z} - \Phi(z) ~ \right| \leq A \, \frac{\log^2 t}{\sqrt{t}},
    \end{equation}
    where $\Phi$ is recalled in \Cref{eq:phi},
    and where we have used the identity~$\sum_{s=1}^t Y_s = N_t(x) - t \, \pi_x$ which follows directly from the definition of~$Y_t$. 
    
    Therefore, if $W_t = \frac{N_t(x) - t \, \pi_x}{\sqrt{\Var(N_t(x))}}$ we can write, for every $N \in \mathbb{N}$
    \begin{align*}
       & |N_t(x) - m| \leq c \sqrt{\Var(N_t(x))} \\ &\iff (N_t(x) - t \, \pi_x) \in \left[ m - t \, \pi_x - c\sqrt{\Var(N_t(x))} , m - t \, \pi_x + c\sqrt{\Var(N_t(x))} \right] \\
        &\iff W_t \in \left[ \frac{m - t \, \pi_x - c\sqrt{\Var(N_t(x))}}{\sqrt{\Var(N_t(x))}} , \frac{m - t \, \pi_x + c\sqrt{\Var(N_t(x))}}{\sqrt{\Var(N_t(x))}} \right],
    \end{align*}
    and we denote the last interval with $ \left[ z_1(t,m) , z_2(t,m) \right]$.
    By \Cref{eq:applied_berry_esseen},
    \begin{equation}
        \bbP \pa{ W_t \in \left[ z_1(t,m) , z_2(t,m) \right] } \leq \pa{\Phi(z_2(t,m)) - \Phi(z_1(t,m))} +2 A \, \frac{\log^2 t}{\sqrt{t}}.
        \label{eq:W_t_interval}
    \end{equation}
    Moreover, it holds that, for any $m \in \mathbb{N}$ we can express the Gaussian integral in terms of the error function $\mathrm{erf}$
    \[\Phi(z_2(t,m)) - \Phi(z_1(t,m)) = \frac{1}{\sqrt{2\pi}}\int_{z_1(t,m)}^{z_2(t,m)} e^{-s^2/2}ds \leq \frac{1}{\sqrt{2 \pi}}\int_{-\ell}^{+\ell} e^{-s^2/2}ds = \mathrm{erf}\pa{\frac{\ell}{\sqrt{2}}}\]
    where $\ell = \frac{1}{2}(z_1(t,m)-z_2(t,m)) = c$, and so $\mathrm{erf}(\ell/\sqrt{2}) := \delta(c)$ that is independent from $m$ and $t$. Hence, it holds that
    \[\sup_{m\in \mathbb{N}}(\Phi(z_2(t,m)) - \Phi(z_1(t,m)))  \leq \delta\]
    which, together with \cref{eq:W_t_interval}, proves that, for sufficiently large $t$, there exists $\varepsilon(c)>0$  such that
    \[\sup_{m \in \mathbb{N}}\Prob{W_t \in [z_1(t,m),z_2(t,m)]} \leq 1-\varepsilon,\]
    which concludes the proof of \Cref{thm:visits_anti_concentration}.
\end{proof}

\begin{lemma}    \label{lem:sum_geometric}
    Let $W_1,\dots,W_n$ be a sequence of i.i.d. geometric random variable with success probability $p$. Then, we have that
    \[\Prob{\sum_{i=1}^{n}W_i \geq k} = \Prob{\mathrm{Bin}(k,p) \leq n}\]
\end{lemma}

\begin{proof}
    Asking that $\sum_{i=1}^n W_i \geq k$ is like asking that, in $k$ Bernoulli trials, we have less than $n$ successes.
\end{proof}

\end{document}